\documentclass[]{IEEEtran}

\usepackage{cite}

\usepackage{tikz} 
\usepackage{pgfplots}
\usepackage{tumcolor}

\usepackage[cmex10]{amsmath}
\usepackage{amssymb}
\usepackage{algorithmicx}
\usepackage{algpseudocode}
\usepackage{algorithm}
\usepackage{booktabs}

\usepackage{array}

\usepackage{amsthm}
\newtheorem{thm}{Theorem}

\newtheorem{cond}{Condition}
\newtheorem{condprime}{Condition}

\newtheorem{lemma}{Lemma}

\usepackage{bm}
\usepackage{dane}

\pgfplotsset{compat=1.12}

\begin{document}

\title{A Bilinear Equalizer for Massive MIMO Systems}

\author{David Neumann, Thomas Wiese, Michael Joham, and Wolfgang Utschick%
        \thanks{The authors are with the Professur f\"ur Methoden der Signalverarbeitung, Technische Universit\"at M\"unchen, 80290 M\"unchen, Germany (email: \{d.neumann, thomas.wiese, joham, utschick\}@tum.de).}%
}

\markboth{}%
{}

% make the title area
\maketitle

\begin{abstract}
    We present a novel approach for low-complexity equalizer design well-suited for cellular massive MIMO systems.
    Our design allows to exploit the channel structure in terms of covariance matrices to improve the performance in the face of pilot-contamination, while basically keeping the complexity of a matched filter.
    This is achieved by restricting the equalizer to functions which are bilinear in the received data signals and the observations from a training phase.
    The proposed design generalizes several previous approaches to equalizer design for massive MIMO.
    We show by asymptotic analysis that with the proposed design the achievable rate grows without bound for growing numbers of antennas even in the presence of pilot-contamination.
    We demonstrate with numerical results that the proposed design is competitive with more complex approaches in a practical cellular setup.
\end{abstract}

%
%\begin{IEEEkeywords}
%\end{IEEEkeywords}

\IEEEpeerreviewmaketitle

\section{Background} 

Recent research into massive MIMO systems, i.e., cellular networks with a large number of antennas at the base stations~\cite{rusek_scaling_2013,marzetta_noncooperative_2010}, which serve a large number of users,
have led to a rediscovery of the dimensionality bottleneck imposed by the fixed coherence interval of the channel~\cite{adhikary_joint_2013}.
Information theoretic results on non-coherent block fading channels~\cite{zheng_communication_2002} indicate that the number of simultaneously transmitted interference free data streams cannot exceed half of the coherence interval.
In scenarios with mobile users, this result severely limits the potential multiplexing gains of a massive MIMO system.

An overloaded system where we serve more users than we have pilot-sequences, violates these conditions. 
We are not able to train all users orthogonally, which leads to interference during the training, so called pilot-contamination~\cite{bjornson_massive_2016,hoydis_massive_2013,jose_pilot_2011}.
Pilot-contamination can have a severe impact on the achievable data rate~\cite{marzetta_noncooperative_2010}.
However, the result in~\cite{zheng_communication_2002} only holds for i.i.d. channel coefficients.
In fact, structure of the channel vectors in form of second-order or subspace information can be exploited to break out of the dimensionality bottleneck~\cite{shariati_low-complexity_2014,yin_robust_2016,bjornson_massive_2017,neumann_cdi_2015,neumann_rate-balancing_2015,adhikary_joint_2013,yin_coordinated_2013}.

In our view, the challenge for receive and transmit filter design in massive MIMO is to reduce the complexity such that it is actually possible to implement those filters in practice, but at the same time exploit structural information to reduce the impact of the limited coherence interval.
The proposed optimal bilinear equalizer (OBE) design achieves the desired trade-off. 
For each fast fading channel realization, the OBE design has a complexity that is similar to the matched filter (MF) based on minimum mean square error (MMSE) estimates of the channel vectors, pushing the more demanding computations into the time scale of the variation of the channel covariance matrices.
Yet, the OBE is able to fully exploit the structural properties captured in the covariance matrices.
This leads to a significant increase in performance compared to the MF.

For ease of exposition of our novel approach, we focus on the cellular uplink.
This could be a cell-free system with distributed antennas~\cite{ngo_cell-free_2015}, but also a large-scale base station in a classical multi-cell network.
In our network, a base-station with $M$ antennas receives signals from $K$ single-antenna terminals that transmit simultaneously.
The $K$ users include the users served by the base station, but may also include additional interfering users from neighboring cells.
We assume flat fading channels with a fixed coherence interval of $T$ channel accesses.
The channel in one coherence interval from user $k$ to the base station, $\vh_k \sim \CN(\zeros, \vC_k)$, is circularly symmetric, complex Gaussian distributed, with zero mean and covariance matrix $\vC_k$.

For the training, $\Ttr < T$ channel accesses are used to transmit pilot symbols.
Specifically, each user transmits one of $\Ttr$ predefined orthogonal pilot sequences.
We assume that there are more users than pilots ($K > \Ttr$).
Consequently, at least one of the training sequences is transmitted by multiple users.

The base station correlates the signals received during the training phase with each of the pilot sequences.
This leads to the least-squares (LS) estimates of the channel vector for user $k$
\begin{align}
    \vobs_k = \vh_k + \sum_{n\in \set I_k} \vh_n + \frac{1}{\sqrt{\ptr}}\vw_k
\end{align}
where $\set I_k$ denotes the set of users that transmit the same pilot sequence as user $k$ (user $k$ itself is not in the set) and $\ptr$ is the equivalent SNR of the training and the noise is normalized such that $\vw_k \sim \CN(\zeros,\id)$.
Note that the LS estimate is identical for users that employ the same pilot sequence, i.e., $\vobs_k = \vobs_n$ if $k$ and $n$ use the same pilot sequence.

The covariance matrix of the LS estimates,
\begin{align}
    \vQ_k = \expec[\vobs_k\vobs_k\he] = \vC_k + \sum_{n\in \set I_k} \vC_n + \frac{1}{\ptr}\id
\end{align}
is simply a sum of the involved channel covariance matrices plus the noise covariance matrix, since we assume that all channel vectors are independent.

In the data-transmission phase, the base station receives signals
\begin{align}
    \vy = \sum_k \sqrt{p_k} \vh_k s_k + \vv
\end{align}
with the independent data symbols $s_k \sim \CN(0,1)$ and normalized transmit power $p_k$ as well as additive white Gaussian noise $\vv \sim \CN(\zeros, \id)$.
We assume separate linear processing for each user.
That is, for each user we filter the received signal $\vy$ with a linear filter $\vg_k$ to get the estimate 
\begin{align}\label{eq:uplink_signal}
    \hat s_k = \sqrt{p_k}\vg_k\he \vh_k s_k + \sum_{n\neq k} \sqrt{p_n}\vg_k\he\vh_n s_n + \vg_k\he\vv.
\end{align}
The transmit signals from the other users are considered as noise.

The MF based on the LS estimate, $\vg_k = \vobs_k$, is regarded as a good match for massive MIMO systems because of its low complexity and satisfactory performance for a large number of antennas due to the asymptotic orthogonality of the channel vectors~\cite{rusek_scaling_2013}.
However, the MF suffers from severe performance degradation in the presence of pilot-contamination~\cite{jose_pilot_2011,marzetta_noncooperative_2010}.
In the literature, several improved filter or precoder designs were proposed which can be interpreted as an adaption of the MF concept to a scenario with imperfect channel state information (CSI).
One example is the MF based on the minimum mean square error (MMSE) estimate of the channel~\cite{yin_coordinated_2013}, i.e.,
\begin{equation}\label{eq:mmse_mf}
    \vg_k = \vC_k \vQ_k\inv \vobs_k.
\end{equation}
However, even when the MMSE estimate is used instead of the LS estimate, the achievable data rate saturates in the presence of pilot contamination.

The MMSE-MF is of the form
\begin{align}\label{eq:gmf_design}
    \vg_k = \vA_k \vobs_k
\end{align}
where $\vA_k$ is a deterministic matrix that only depends on the channel statistics.
We refer to this kind of filter with an arbitrary $\vA_k$ as a bilinear equalizer (BE).
The reason is that, for a BE, the estimate 
\begin{align}\label{eq:bilinear}
    \hat s_k = \vobs_k\he \vA_k\he \vy
\end{align}
is not only linear in the data vector $\vy$, which is a common restriction, but also linear in the observations $\vobs_k$.

In other words, we are designing a bilinear estimator for the data symbols that only depends on the channel statistics.
Instead of choosing an ad-hoc design for the BE, such as the MMSE channel estimate in~\eqref{eq:mmse_mf}, we will show how to calculate the optimal transformations $\vA_k$ with respect to a lower bound on the achievable rates.
This leads to an optimal BE (OBE) -- in terms of this lower bound -- that is able to suppress the interference caused by pilot contamination as the number of antennas grows large.

The calculations of the filters $\vg_k$ need only one matrix-vector multiplication per user.
The complexity of the multiplication depends on the structure of the matrix $\vA_k$.
We will learn later how to exploit structure of the covariance matrices for common array geometries to reduce the complexity from $O(M^2)$ to $O(M)$.
That is, we can reduce the total complexity of the filter calculations to $O(MK)$ floating point operations.
In comparison, for a (regularized) zero-forcing filter, such as the linear MMSE (LMMSE) filter discussed in~\cite{neumann_mse_2017,bjornson_massive_2017}, the complexity is $O(MK^2)$ under the same assumptions and with a larger constant factor.

A special case of the OBE approach for a network MIMO setup and a channel model with rich scattering has been discussed in~\cite{adhikary_uplink_2017,ashikhmin_pilot_2012}.
We comment on the connection in Section~\ref{sec:structure}.

In summary, this work contains the following contributions.
\begin{itemize}
    \item We introduce the BE as an approach to low-complexity filter design for a scenario with imperfect CSI.
    \item We derive the optimal BE with respect to a lower bound on the achievable rates, denoted as OBE.
    \item We show that for certain structure of the covariance matrices, the complexity of calculating the OBE reduces significantly and discuss how this result applies to practically relevant array geometries.
    \item Under mild conditions on the covariance matrices, we show that the OBE achieves asymptotically linear scaling of the SINR with respect to the number of antennas, even in the presence of pilot-contamination.
    \item We compare the asymptotic result with the asymptotically achievable SINR of the more complex linear MMSE filter.
    \item We show that even with a non-optimal BE and only partial knowledge of the covariance matrices, we can achieve asymptotically linear scaling of the SINR with respect to the number of antennas. 
    \item We showcase the performance of optimal OBE and sub-optimal BE design in comparison with state-of-the-art methods via numerical simulations.
\end{itemize}

\subsection*{Notation and Definitions}

For a matrix $\vX$, we denote the transpose as $\vX\tp$ and the conjugate transpose as $\vX\he$.
The Frobenius norm is denoted by $\norm{\vX}_F$ and the spectral norm by $\norm{\vX}_2$.
The operator $\vec(\vX)$ yields the vector with the stacked columns of $\vX$.

The matrix $\id$ denotes the identity matrix and the vector $\ve_k$ denotes the $k$-th canonical unit vector.

We only consider limits where the number of antennas $M$ goes to infinity.
In this context, asymptotic equivalence, of two sequences $a_M$ and $b_M$, denoted as $a_M \asymp b_M$, means that
\begin{align}
    \lim_{M\rightarrow\infty} a_M/b_M = 1.
\end{align}
If we know, e.g., that $\liminf_{M\rightarrow \infty} b_M/M > 0$, we have $b_M \asymp a_M$ iff
\begin{align}
    \lim_{M\rightarrow\infty} a_M/M - b_M/M = 0.
\end{align}
For matrices and vectors with fixed dimension, asymptotic equivalence means element-wise asymptotic equivalence.

For matrices that grow with the number of antennas, we need a different, ``weak'', notion of asymptotic equivalence~(cf.~\cite{gray_toeplitz_2006}). 
Two matrices $\vA\in \C^{M\times M}$ and $\vB\in \C^{M\times M}$ are asymptotically equivalent, denoted by $\vA\asymp_w\vB$, if
\begin{align}
    \lim_{M\rightarrow \infty} \frac{1}{M} \norm{\vA - \vB}_F = 0.
\end{align}

\section{Optimal Bilinear Equalizer}
\label{sec:uplink}

In this section, we discuss a design of the transformation $\vA_k$ in~\eqref{eq:gmf_design}, that maximizes a lower bound on the achievable rates.
We then show in the next section how we can exploit structure of the covariance matrices to reduce the computational complexity.
Finally, we show in Section~\ref{sec:asymptotic} that when we use these optimal transformations, the achievable rate grows without bound for $M$ going to infinity.

The signal model for a user $k$ in~\eqref{eq:uplink_signal} is equivalent to a SISO channel
\begin{align}
    \hat s_k &= \underbrace{\sqrt{p_k}\vg_k\he \vh_k}_{h_\text{eff}} s_k + \underbrace{\sum_{n\neq k} \sqrt{p_n}\vg_k\he\vh_n s_n + \vg_k\he\vv}_{v_\text{eff}}\\
    &= h_\text{eff} s_k + v_\text{eff}.
\end{align}
Because of the imperfect CSI, we cannot provide a closed form expression for the mutual information $I(s_k;\hat s_k)$.
Instead we use a lower bound that is often found in massive MIMO literature (e.g.~\cite{bjornson_massive_2017,bjornson_massive_2016}), which is based on the worst-case noise bound in~\cite{medard_effect_2000,hassibi_how_2003}.
For the bound, we need the MMSE estimate of the effective channel $\hat h_\text{eff} = \expec[h_\text{eff}]$.
Here, we ignore any side information that we have on the channel and thus the MMSE estimate is simply the expectation.
We get a different bound if we condition the expectation on the available observations. 
We discuss the bound that uses the conditional expectations in the context of the LMMSE filter in Appendix~\ref{app:lmmse}.

With the estimate $\hat h_\text{eff}$, the lower bound evaluates to
\begin{align}
    I(s_k;\hat s_k) \geq \log_2(1+ \gamma_k\ul)
\end{align}
with the SINR expression
\begin{align}
    \gamma_k\ul &= \frac{\lvert\hat h_\text{eff}\rvert^2}{ \var(\hat h_\text{eff}) + \var(v_\text{eff}) } \\
                &= \frac{ p_k\abs{\expec[\vg_k\he\vh_k]}^2}{  p_k \var(\vg_k\he\vh_k) + \sum_{n\neq k} p_n \expec[\abs{\vg_k\he\vh_n}^2] + \expec[\vg_k\he\vg_k] }.\label{eq:sinr_lower_bound}
\end{align}

For our BE approach with $\vg_k = \vA_k \vobs_k$, the channel vectors and the filters are jointly Gaussian distributed and thus, the expectations can be calculated analytically leading to~(cf.~\cite{bjornson_massive_2016-1})
\begin{equation}\label{eq:uplink_sinr_gmf}
\gamma_k\ul = \frac{p_k\abs{\tr(\vC_k \vA_k)}^2}{ \tr(\vZ\vA_k\vQ_k\vA_k\he) + \sum_{n \in \set I_k} p_n \abs{ \tr(\vC_n\vA_k)}^2 }
\end{equation}
with the covariance matrix 
\begin{align}
    \vZ = \id + \sum_n p_n \vC_n
\end{align}
of the received signal $\vy$.
For the following analysis, we replace the transformations with their vectorized form $\va_k = \vec(\vA_k) \in \mathbb C^{M^2}$ which leads to
\begin{equation}\label{eq:uplink_sinr_gmf_vec}
    \gamma_k\ul = \frac{p_k\abs{\vc_k\he\va_k}^2}{ \va_k\he (\vQ_k\tp\otimes\vZ) \va_k + \sum_{n \in \set I_k} p_n \abs{ \vc_n\he\va_k}^2 }.
\end{equation}
where $\vc_k = \vec(\vC_k)$ and we used the fact that 
\begin{align}
    \tr(\vA_k\he \vZ \vA_k \vQ_k) = \va_k \he \vec(\vZ \vA_k \vQ_k) = \va_k\he (\vQ_k\tp \otimes \vZ) \va_k.
\end{align}
The summation over the interfering users $\set I_k$ in the denominator of~\eqref{eq:uplink_sinr_gmf_vec} represents the interference caused by pilot-contamination, which is additional to the general interference that appears even when all users are trained with orthogonal training sequences.

If we compare the SINR in~\eqref{eq:uplink_sinr_gmf_vec} with the instantaneous SINR for perfect CSI at the receiver (with general noise covariance matrix $\vSigma$),
\begin{align}
    \gamma_k^\text{CSI} = \frac{\abs{\vg_k\he\vh_k}^2}{\vg_k\he\vSigma\vg_k + \sum_n \abs{\vg_k\he \vh_n}^2}
\end{align}
we note that the structure is equivalent and only the dimension of the vectors is increased from $M$ to $M^2$.
The linear filters are replaced by the vectorized transformations and the channel vectors are replaced by the channel covariance matrices.

The structural similarity allows us to apply various techniques from the MIMO literature within the BE framework.
For example, we can use the well-known uplink-downlink SINR duality to design BE precoders for downlink transmission (cf. also~\cite{bjornson_massive_2016}).
We can also apply methods from the literature for power allocation, network utility maximization, etc.
Since the SINRs in~\eqref{eq:uplink_sinr_gmf} only depend on the channel statistics, all of the various optimizations for resource allocation only have to be done in the time scale at which the covariance matrices change.

Also analogously to the case of perfect CSI, we can explicitly calculate the optimal transformations in the uplink.
Since $\gamma_k\ul$ only depends on the transformation $\va_k$ and not on the $\va_n, n\neq k$, we are able to calculate the optimizer
\begin{align}
    \va_k^\star &= \argmax_{\va_k} \gamma_k\ul \notag \\ 
                &= ( \vQ_k\tp\otimes\vZ + \sum_{n\in \set I_k} p_n \vc_n \vc_n\he )\inv \vc_k \label{eq:opt_uplink_trafo}
\end{align}
with the corresponding optimal SINR
\begin{align}\label{eq:uplink_sinr_gmf_opt}
    \gamma_k^\star &=  p_k \vc_k\he ( \vQ_k\tp\otimes\vZ + \sum_{n\in \set I_k} p_n \vc_n \vc_n\he )\inv \vc_k\\
                   &= p_k \vc_k\he \va_k^\star.
\end{align}

We refer to the BE $\vg_k = \vA_k^\star \vobs_k$ that is calculated using the optimal transformation $\vec(\vA_k^\star) = \va_k^\star$ as OBE.
We will see in Section~\ref{sec:asymptotic} that the OBE is able to deal with pilot-contamination in the sense that the SINR does not saturate for large numbers of antennas.
The OBE can be interpreted as an adaption of the classical MF to the case of imperfect CSI by using a statistical pre-filter.

\subsection*{SINR Reformulation}

For asymptotic analysis and efficient numerical implementation, we reformulate the optimal SINRs in~\eqref{eq:uplink_sinr_gmf_opt} and the optimal transformations $\va_k^\star$ in~\eqref{eq:opt_uplink_trafo} into a more convenient form.
To this end, we consider the set of users $\Omega_p = \{1,\ldots, K_p\}$ which use the pilot sequence $p$.
Remember that the observation $\vobs_k$ is the same for all users $k \in \Omega_p$ and has the covariance matrix $\vQ_1 = \ldots=\vQ_{K_p} = \vQ$ (we drop the user index for notational convenience).

We collect the vectorized covariance matrices into
\begin{equation}
    \vXi = [\vc_1, \ldots, \vc_{K_p}] \in \mathbb C^{M^2 \times K_p}.
\end{equation}
The matrix $\vXi$ basically takes the role that the channel matrix has for perfect CSI.

To give a succinct reformulation of the SINR we additionally define the matrix $\vGamma \in \C^{K_p\times K_p}$ with the elements
\begin{equation}\label{eq:gamma_elementwise}
    [\vGamma]_{nk} = \tr(\vC_n \vZ\inv \vC_k \vQ\inv).
\end{equation}
Using the Kronecker-product vectorization trick, this matrix can be also stated as
\begin{equation}\label{eq:gamma}
    \vGamma = \vXi\he (\vQ^{-\Tr}\otimes \vZ\inv) \vXi/M
\end{equation}
which is the form we will use for the asymptotic analysis in Section~\ref{sec:asymptotic}.

\begin{lemma}\label{thm:sinr_reformulation}
    For $p_k > 0\;\forall k$ and with $\vP = \diag(p_1,\ldots,p_{K_p})$, the uplink SINRs in~\eqref{eq:uplink_sinr_gmf_opt} can be equivalently stated as 
    \begin{equation}\label{eq:uplink_sinr_gmf_reform}
        \gamma_k^\star = M p_k\frac{\ve_k\tp \vGamma (\frac{1}{M}\vP\inv + \vGamma)\inv \ve_k }{\ve_k\tp (\frac{1}{M}\vP\inv + \vGamma)\inv \ve_k}.
    \end{equation}
    The optimal transformations can be alternatively expressed as
    \begin{align}\label{eq:mmse_trafo}
        \tilde\va_k^\star = \frac{1}{p_k} (\vQ\tpi \otimes \vZ\inv) \vXi \big(\vP\inv + \vGamma\big)\inv \ve_k
    \end{align}
    which is a scaled version of $\va_k^\star$ in~\eqref{eq:opt_uplink_trafo}.
\end{lemma}
\begin{proof}
    See Appendix~\ref{app:sinr_reformulation}.
\end{proof}

Note that the transformations in~\eqref{eq:mmse_trafo} are scaled versions of the optimal transformations in~\eqref{eq:opt_uplink_trafo}.
But since scaling does not effect the SINR in~\eqref{eq:uplink_sinr_gmf}, both choices are optimal.
In fact, in analogy to the LMMSE filter, which minimizes the MSE and maximizes a bound on the mutual information at the same time, the transformations in~\eqref{eq:mmse_trafo} also minimize the MSE $\expec[\lvert s_k - \vobs_k\he\vA_k\he\vy\rvert^2]$~\cite{neumann_mse_2017}.

The complexity of calculating the transformations $\vA_k^\star$ is dominated by the calculation of the matrix $\vGamma$.
An efficient way to calculate $\vGamma$ is to first calculate $\vZ\inv\vC_k\vQ\inv$ for all $k\in\Omega_p$ and then calculate the inner products with all $\vC_n$.
This procedure leads to a complexity of $\order(M^3 K)$ floating point operations.
For general covariance matrices, the calculation of the OBEs $\vg_k = \vA_k^\star \vobs_k$ needs $\order(M^2K)$ operations in total.
Thus, as long as the coherence interval of the covariance matrices in terms of channel coherence intervals is much larger than the number of antennas, the complexity of the calculation of the $\vA_k^\star$ is negligible.
The calculation of the LMMSE filter also needs $\order(M^2K)$ operations, but with a larger constant factor, since several matrix-vector multiplications are required per user.

In literature, different assumptions appear regarding the coherence interval of the covariance matrices, ranging from 40 channel coherence intervals~\cite{adhikary_uplink_2017} to 25000~\cite{bjornson_massive_2016-1}.
The smaller number is not larger than the number of antennas in a typical massive MIMO system.
For this reason, but also because of the high complexity of the filter calculations, we discuss in the following how to reduce the computational complexity for common antenna array geometries.

\section{Exploiting Array Structure}
\label{sec:structure}

To reduce the complexity of the calculations of the transformations $\vA_k^\star$, but also the calculation of the OBEs $\vg_k = \vA_k^\star \vobs_k$, we exploit common structure of the covariance matrices.

We notice that the vectorized optimal transformations $\va_k^\star$ in~\eqref{eq:mmse_trafo} are linear combinations of vectors $(\vQ\tpi \otimes \vZ\inv) \vc_n$.
Reverting the vectorization, we see that 
\begin{equation}\label{eq:transformation_structure}
    \vA_k^\star =  \vZ\inv \left( \sum_{\ell \in \Omega_p} \sigma_{k\ell} \vC_\ell \right) \vQ\inv
\end{equation}
for some $\sigma_{k\ell}$.
For the following discussion it is important to remember the structure of
\begin{align}
    \vZ = \id + \sum_k p_k \vC_k,
\end{align}
and
\begin{align}
    \vQ = \frac{1}{\ptr}\id + \sum_{k\in\Omega_p} \vC_k.
\end{align}

Taking a close look at~\eqref{eq:transformation_structure}, we realize that certain structure of the covariance matrices carries over to the transformations $\vA_k^\star$. 
An important example are covariance matrices that share the same eigenbasis, i.e.,
\begin{align}
    \vC_k = \vU \diag(\hat\vc_k) \vU\he \;\;\forall k
\end{align}
for some unitary $\vU$.
We can easily verify that the matrices $\vZ$ and $\vQ$ and thus the optimal transformations have the same eigenbasis as well, i.e., we have 
\begin{align}
    \vA_k^\star = \vU \diag(\hat\va_k^\star) \vU\he.
\end{align}
Therefore, if the channel covariance matrices have the desired structure, we can simply transform the incoming signals $\vy$ and $\vobs_k$ by $\vU\he$ (cf.~\eqref{eq:bilinear}).
Then we no longer have to consider the eigenbasis $\vU$, but we can work with the diagonal matrices containing the eigenvalues $\hat \vc_k$ as if we had diagonal channel covariance matrices (and in fact we have diagonal covariance matrices with respect to the basis $\vU$).

If the covariance matrices are diagonal, the matrices $\widehat\vZ = \diag(\hat\vz)$ and $\widehat\vQ = \diag(\hat\vq)$ are diagonal as well.
The operation which dominates the complexity of calculating the transformations $\vA_k^\star$ in~\eqref{eq:mmse_trafo} is still the calculation of the matrix $\vGamma$.
For diagonal covariance matrices, the computational complexity reduces to $\order(M\sum_p K_p^2)$.
If the number of users $K_p$ that transmit a pilot sequence $p$ is the same for all pilot sequences, i.e., $K_p = K/\Ttr$, we have a complexity of $\order(MK^2/\Ttr)$.
This is lower than the $\order(MK^2)$ operations required for the LMMSE filter \emph{in each channel coherence interval} for diagonal covariance matrices.
The complexity to calculate all OBEs $\vg_k = \vA_k^\star \vobs_k$ reduces to $\order(MK)$ which is significantly lower than the complexity of the LMMSE filter.
We see that even if the covariance matrix is only constant for a relatively small number of channel coherence intervals, the calculation of the transformations $\vA_k^\star$ does not significantly effect the total computational complexity.

One major advantage of structured covariance matrices, besides the gain in computational complexity, is the simplified estimation of the second order statistics.
With the common eigenbasis we only have one parameter per spatial direction and the parameters can be estimated independently~(cf.~\cite{neumann_low-complexity_2015}).
Thus, for the case of estimated covariance matrices, the assumption of a certain structure might actually improve the performance compared to the general case without any assumptions.

In the following, we give several examples that exhibit special structures of the channel covariance matrices that carrys over to the optimal transformations $\vA_k^\star$.

\subsection{Cell-free Scenario}

For distributed antennas~\cite{ngo_cell-free_2015,ngo_cell-free_2017}, we typically have diagonal covariance matrices, i.e., $\vU = \id$.
Consequently we have diagonal transformations $\vA_k$.
The calculation of the linear filters reduces to an element-wise multiplication $\vg_k = \hat\va_k \odot \vobs_k$ with linear complexity in the number of antennas.
As a side-effect of the diagonal transformations $\vA_k$, it is no longer necessary to collect the instantaneous observations $\vobs_k$ at a central hub.
We only need to send the symbols $s_k$ for the users to all antennas which can then use the local observations to calculate the transmit signal.
The estimated variances at the different antennas have to be collected, which requires a much lower overhead than collecting instantaneous CSI.

This special case appears in similar form in previous work~\cite{ngo_cell-free_2017}.
In contrast to the BE design, the authors interpret the coefficients in $\hat\va_k$ as power allocation and thus only allow positive values. 
With this restriction the complete elimination of pilot-contamination in the asymptotic limit is no longer possible.
Our approach achieves asymptotically optimal performance in a cell-free network with the same requirements on computational complexity and communication overhead as the approach in~\cite{ngo_cell-free_2017}.

\subsection{Network MIMO}

In a network MIMO scenario, all base stations in (a neighborhood of) the network jointly process the received signals.
If we have the unfavorable case that the covariance matrices of the channels from the users to one base station are scaled identity matrices, we get full covariance matrices of the form
\begin{align}
    \vC_k = \blkdiag(\beta_{1k} \id, \ldots, \beta_{Lk}\id)
\end{align}
where $L$ is the number of base stations in the network.
If we look again at \eqref{eq:transformation_structure}, we note that this kind of block structure also transfers to the transformations $\vA_k^\star$.
That is, the transformations $\vA_k^\star$ will also have the block identity structure.

This scenario is similar to the cell-free case, however there are only $L$ different coefficients for each user.
As a consequence, at most $L$ users can use the same pilot sequence if we want the covariance matrices to be linearly independent (this is required for the asymptotic optimality we show in the next section).
The OBE approach for this special covariance matrix structure was discussed (under a different name) in~\cite{adhikary_uplink_2017,ashikhmin_pilot_2012}.
The authors in~\cite{adhikary_uplink_2017} focus on methods for optimal uplink power allocation with respect to the max-min criterion.

\subsection{Uniform Arrays}

For uniform linear and uniform rectangular arrays, the covariance matrices can be approximately diagonalized by the DFT matrix or a Kronecker-product of DFT matrices.
For uniform linear arrays, this approximation is accurate for larger numbers of antennas (cf. Appendix~\ref{app:channel_model}) and, thus, quite popular in the massive MIMO literature~\cite{adhikary_joint_2013}.
Thanks to the fast Fourier transform (FFT), the complexity of applying the transformation $\vU\he$ to the incoming signals reduces to $O(M \log M)$.
Since the incoming and outgoing signals have to be transformed via the FFT, the processing needs to be centralized in this case.

\section{Asymptotic Analysis}
\label{sec:asymptotic}

With the reformulation in Lemma~\ref{thm:sinr_reformulation} the analysis of the asymptotic behaviour of the OBE is quite straightforward.
For large $M$ and $\vp > 0$, the matrix $\vP\inv/M$ vanishes.
Thus, the asymptotic behaviour of $\gamma_k^\star$ depends critically on the asymptotic behaviour of $\vGamma$, specifically, the rank of $\vGamma$.
Formally we can say:
\begin{thm}\label{thm:asymptotic}
    If $\vp> 0$ and $\limsup_{M\rightarrow \infty} \lVert \vGamma\inv \rVert_2 < \infty$ then $\gstar_k \asymp \gasy_k$ where
    \begin{align}\label{eq:uplink_sinr_asy}
        \gamma_k^\text{asy} = M \frac{p_k}{\ve_k\tp \vGamma\inv \ve_k} \quad\forall k\in\Omega_p.
    \end{align}
    Additionally, we have
     \begin{align}
         \liminf_{M\rightarrow \infty} \gasy_k/M > 0 \quad \forall k \in \Omega_p.
     \end{align}
\end{thm}
\begin{proof}
    Follows straightforward from Lemma~\ref{thm:sinr_reformulation} by taking the limit of the SINR in~\eqref{eq:uplink_sinr_gmf_reform} for $M\rightarrow \infty$.
\end{proof}

Now the question is whether we can find a set of intuitive conditions on the channel covariance matrices $\vC_k$ that guarantee the assumption $\limsup_{M\rightarrow \infty} \lVert \vGamma\inv\rVert_2 < \infty$.
One necessary condition is that the captured energy grows linearly with the number of antennas.
\begin{condprime}\label{cond:growing_trace}
    \begin{align}
        \liminf_{M\rightarrow\infty} \tr(\vC_k)/M > 0, \quad \forall k.
    \end{align}
\end{condprime}
Without this condition, the diagonal elements of $\vGamma$ might vanish, which would violate the assumption in Theorem~\ref{thm:asymptotic}.
Commonly used channel models that ignore antenna coupling fulfill Condition~\ref{cond:growing_trace}.
If antenna coupling is taken into account, Condition~\ref{cond:growing_trace} can only be achieved with an array aperture that grows linearly with $M$~\cite{ivrlac_physical_2016}.

Another necessary condition is that the columns of $\vXi$ -- which correspond to the covariance matrices $\vC_k$ of users with common pilot sequence -- are asymptotically linearly independent.
This can be written as~(cf.~\cite{bjornson_massive_2017})
\begin{cond}\label{cond:linear_independence}
    \begin{align}\label{eq:linear_independence1}
        \liminf_{M\rightarrow\infty} \min_{\vlambda: \lVert \vlambda \rVert_2 = 1} \frac{1}{M} \bigg\lVert \sum_{k\in\Omega_p} \lambda_k \vC_k \bigg\rVert_F^2 > 0
    \end{align}
    or, equivalently,
    \begin{align}\label{eq:linear_independence2}
        \limsup_{M\rightarrow\infty} \lVert (\vXi\he\vXi/M)\inv\rVert_2 < \infty.
    \end{align}
\end{cond}
The equivalence of~\eqref{eq:linear_independence1} and~\eqref{eq:linear_independence2} follows from 
\begin{align}
    \min_{\vlambda: \lVert \vlambda \rVert_2 = 1} \frac{1}{M} \bigg\lVert \sum_{k\in\Omega_p} \lambda_k \vC_k \bigg\rVert_F^2 &= \min_{\vlambda:\lVert \vlambda \rVert_2 = 1}\vlambda\he (\vXi\he \vXi/M) \vlambda \notag \\
                                                                                                                              &= \sigma_{K_p}(\vXi\he\vXi/M) \notag \\
                                                                                                                              &= \frac{1}{\norm{(\vXi\he\vXi/M)\inv}_2}
\end{align}
where $\sigma_{K_p}(\vXi\he\vXi/M)$ is the smallest singular value of $\vXi\he\vXi/M$.
Condition~\ref{cond:linear_independence} actually implies Condition~\ref{cond:growing_trace}.
In Appendix~\ref{app:channel_model}, we show that for a uniform linear array (ULA), Condition~\ref{cond:linear_independence} is equivalent to having linearly independent angular power densities.
In practical channel models (e.g.~\cite{3gpp_spatial_2014}), if users are at slightly different positions in the network, their angular power densities are linearly independent.
Thus, we believe that Condition~\ref{cond:linear_independence} is not very restrictive.

We need an additional condition to deal with the weighting matrix $\vQ\tpi \otimes \vZ\inv$ in $\vGamma = \vXi\he (\vQ\tpi\otimes \vZ\inv)\vXi$.
Since both, $\vQ$ and $\vZ$, contain sums of covariance matrices, one possible condition is
\begin{cond}\label{cond:bounded_norm}
    \begin{align}
        \limsup_{M\rightarrow\infty} \lVert \vC_k \rVert_2 < \infty, \quad \forall k
    \end{align}
\end{cond}
This condition basically forces the energy in the channel vector to be spread in many spatial directions as $M$ grows.
For example, line-of-sight channel models with rank-one covariance matrices do not satisfy this condition.
With these two conditions, we can show the following lemma, which then implies Theorem~\ref{thm:asymptotic}.
\begin{lemma}\label{lemma:general_conditions}
    If the covariance matrices fulfill Conditions~\ref{cond:linear_independence} and~\ref{cond:bounded_norm}, we have $\limsup \lVert \vGamma\inv\rVert_2 < \infty$.
\end{lemma}
\begin{proof}
    Since the covariance matrices have bounded spectral norm by Condition~\ref{cond:bounded_norm}, we have $\limsup_M \lVert \vZ \rVert_2 < \infty$ and $\limsup_M \lVert \vQ\rVert_2 < \infty$ and thus
\begin{align}
    \limsup_M \lVert \vGamma\inv\rVert_2 \leq \limsup_M \lVert \vZ \rVert_2 \lVert \vQ \rVert_2 \lVert (\vXi\he\vXi/M)\inv \rVert_2 < \infty
\end{align}
due to Condition~\ref{cond:linear_independence}.
\end{proof}

From the result of Lemma~\ref{lemma:general_conditions} we conclude that the SINR of all users grows linearly with $M$ under very general conditions on the channel covariance matrices.
The same observation was made in~\cite{bjornson_massive_2017} for the massive MIMO uplink when a more complex linear MMSE filter is applied.
Note that for the very popular far field channel model with a uniform linear array at the base station (cf. Appendix~\ref{app:channel_model}), Condition~\ref{cond:bounded_norm} is violated.
However, as we also discuss in Appendix~\ref{app:channel_model}, the bounded-norm condition is not necessary for this specific channel model and might also be too restrictive in other cases.

The slope of the SINR depends on the matrix $\vGamma$, which is clearly negatively affected by the correlations between the covariance matrices.
The correlation of the covariance matrices of users that employ the same pilot sequences can be reduced by proper allocation of pilots to users~(cf.~\cite{yin_coordinated_2013}).

If we use an LMMSE receive filter, we get a very similar result for the asymptotically equivalent SINR (cf.~Appendix~\ref{app:lmmse}).
The only difference is that the channel covariance matrices in $\vZ$ are replaced by the covariance matrices of the estimation errors when we apply an MMSE channel estimator.
As a result, we get a matrix that is smaller in terms of positive definiteness and thus a higher slope for the SINR.
Maybe not surprisingly, the OBE and the LMMSE filter exhibit similar performance at low SNR.
An in-depth discussion can be found in Appendix~\ref{app:lmmse}.

\section{Partial Covariance Matrix Information}

In the Section~\ref{sec:structure}, we discussed how common structure of the channel covariance matrices helps to reduce the complexity of calculating the OBE.
We can turn this around and assume a certain structure to reduce complexity, even if the actual covariance matrices do not exactly match that assumption.
An equivalent analysis was performed in~\cite{bjornson_massive_2017} for the LMMSE filter.
For example, we can just use the diagonals $\hat\vc_k$ of the covariance matrices $\vC_k$ (possibly after transformation to a different basis) to design the transformations $\vA_k$. 
We can show that, if the diagonal matrices $\widehat\vC_k = \diag(\hat\vc_k)$ fulfill Conditions~\ref{cond:linear_independence} and~\ref{cond:bounded_norm}, we still get a linear scaling of the SINR with the number of antennas, even if the actual covariance matrices $\vC_k$ are not diagonal.

To see this, let us define the matrix $\widehat \vXi= [\hat \vc_1, \ldots, \hat \vc_{K_p}] \in \R^{M\times K_p}$.
If we assume that the covariance matrices are diagonal, the optimal transformation is diagonal as well, with the diagonal
\begin{align}
    \hat\va_k = \widehat\vZ\inv\widehat\vQ\inv \widehat \vXi ( \vP\inv + \widehat\vXi\tp \widehat\vZ\inv\widehat\vQ\inv \widehat \vXi )\inv \ve_k.
\end{align}
However, due to the mismatch between assumption and reality, $\vA_k = \diag(\hat\va_k)$ is not the optimal transformation. 
Consequently, we can not use the SINR in~\eqref{eq:uplink_sinr_gmf_reform} to evaluate the performance.

Since the transformation is not optimal anyway, it does not affect the analysis if we replace the diagonal matrices $\widehat\vZ\inv\widehat\vQ\inv$ and $\vP\inv$. 
That is, we can use the transformation $\vA_k = \diag(\hat \va_k)$ with
\begin{align}\label{eq:diagonal_trafo}
    \hat\va_k = \vD \widehat \vXi ( \vR + \widehat\vXi\tp \vD \widehat \vXi )\inv \ve_k
\end{align}
for some diagonal matrix $\vD \in \R^{M\times M}$ and a positive definite regularization matrix $\vR\in \R^{K_p \times K_p}$.

We state the following Theorem.
\begin{thm}\label{thm:asymptotic2}
    Suppose we calculate diagonal transformations $\vA_k$, with the diagonal as in~\eqref{eq:diagonal_trafo}, only based on the diagonal elements $\hat \vc_k$ of the covariance matrices. 
    If the matrices $\widehat \vC_k = \diag(\hat \vc_k)$ fulfill Conditions~\ref{cond:linear_independence} and~\ref{cond:bounded_norm} and the diagonal matrix $\vD$ fulfills $\liminf \lVert \vD \rVert_2 > 0$ and $\limsup \lVert \vD\inv \rVert_2<\infty$ we get
    \begin{align}
        \liminf_{M\rightarrow \infty} \gul_k/M > 0.
    \end{align}
\end{thm}
\begin{proof}
    Since a scaling of the transformation does not affect the SINR, we will use 
    \begin{align}\label{eq:scaled_diagonal_trafo}
        \hat\va_k = \frac{1}{\sqrt{M}} \vD \widehat \vXi \left( \frac{1}{M}\vR + \frac{1}{M}\widehat\vXi\tp \vD \widehat \vXi \right)\inv \ve_k
    \end{align}
    in the following.
    The matrix $\vR$ does not grow with $M$ and is considered constant.
    We have 
    \begin{align}\label{eq:asymp_diag}
        &\limsup_{M\rightarrow \infty} \lVert (\widehat \vXi\tp \vD \widehat\vXi)\inv /M \rVert_2 \notag \\
        &\leq \limsup_{M\rightarrow \infty} \lVert \vD\inv \rVert_2 \lVert (\widehat\vXi\tp\widehat\vXi/M)\inv\rVert_2 \\
        &< \infty \notag
    \end{align}
    due to the conditions on $\vD$ and the asymptotic linear independence from Condition~\ref{cond:linear_independence}.

    If we incorporate the diagonal transformation $\vA_k = \diag(\hat\va_k)$ into the uplink SINR in~\eqref{eq:uplink_sinr_gmf}, we get
    \begin{equation}
        \gamma_k\ul/M = \frac{p_k\abs{\hat\vc_k\tp \hat\va_k/\sqrt{M}}^2}{ \hat\va_k\tp (\vQ\tp \odot \vZ) \va_k + \sum_{n \in \set I_k} p_n \abs{ \hat\vc_n\hat\va_k }^2 }
    \end{equation}
    where $\odot$ denotes the element-wise or Hadamard product.
    Note that $\vQ$ and $\vZ$ are not diagonal, since the covariance matrices $\vC_k$ are not actually diagonal, but still $\tr(\vC_k \vA_n) = \hat\vc_k\tp \hat\va_n$.

    Again, due to the conditions on $\vD$ and the bounded spectral norm of the covariance matrices (Condition~\ref{cond:bounded_norm}), it is straightforward to verify that 
    \begin{align}
        \limsup_{M\rightarrow \infty} \hat\va_k\tp (\vQ\tp \odot \vZ) \va_k < \infty.
    \end{align}
    Additionally, we have
    \begin{align}
        &\lim_{M\rightarrow \infty} \hat\vc_k\tp \hat\va_k/\sqrt{M} \\
        &= \lim_{M\rightarrow\infty} \ve_k\tp \widehat\vXi\vD\widehat\vXi/M (\vR/M + \widehat\vXi\vD\widehat\vXi/M)\inv \ve_k \\
        &= 1
    \end{align}
    due to~\eqref{eq:asymp_diag}.

    In the following equations we use $\vGamma_\vD = \widehat\vXi \vD \widehat\vXi/M$ for notational convenience.
    For the interference caused by pilot contamination we get for $n\in \set I_k$
    \begin{align}
        &\lim_{M\rightarrow \infty} \hat\vc_k\tp \hat\va_n \\
        =&\lim_{M\rightarrow\infty} \ve_k\tp \sqrt{M} \vGamma_\vD (\vR/M + \vGamma_\vD)\inv \ve_n \\
        =&\lim_{M\rightarrow\infty} \ve_k\tp \left( \sqrt{M}\vGamma_\vD (\vR/M + \vGamma_\vD)\inv - \sqrt{M} \id \right) \ve_n \\
        =&\lim_{M\rightarrow\infty} \ve_k\tp (\sqrt{M}\vGamma_\vD - \vR/\sqrt{M} - \sqrt{M}\vGamma_\vD )(\vR/M + \vGamma_\vD)\inv \ve_n \notag\\ 
        =&\lim_{M\rightarrow\infty} -\frac{1}{\sqrt{M}}\ve_k\tp  \vR (\vR/M + \vGamma_\vD)\inv \ve_n \\ 
        =& \;0
    \end{align}
    due to~\eqref{eq:asymp_diag}.
    In the end we get the asymptotic equivalence
    \begin{align}
        \gamma_k\ul \asymp M\frac{p_k}{ \ve_k\tp \vGamma_\vD\inv \widehat\vXi\tp \vD (\vZ\odot \vQ\tp) \vD \widehat\vXi \vGamma_\vD\inv \ve_k\tp }
    \end{align}
    and thus 
    \begin{align}
        \liminf_{M\rightarrow \infty} \gul_k/M > 0.
    \end{align}
\end{proof}
As mentioned above, an equivalent result was shown in~\cite{bjornson_massive_2017} for the LMMSE filter.
The result for the LMMSE filter can be proven with the same techniques using a generalized version of the law of large numbers (cf.~Appendix~\ref{app:lmmse}).

Note that we can transfer the result in Theorem~\ref{thm:asymptotic2} to the multi-cell downlink in a straightforward manner.
In contrast to~\cite{bjornson_massive_2017}, we do not need a different bound for the downlink, since we already employ the looser bound in the uplink where the decoder does not use any instantaneous CSI.
%To not distract from the key results in this work, we leave this as an exercise to the reader.

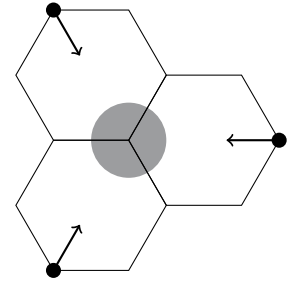
\begin{figure}[h]
    \centering
    \begin{tikzpicture}[]
        \fill[TUMMediumGray] circle(0.5);
        \draw[] (1,0) +(0:1) -- +(60:1) -- +(120:1) -- +(180:1) -- +(240:1) -- +(300:1) -- cycle;
        \draw[] (120:1) +(0:1) -- +(60:1) -- +(120:1) -- +(180:1) -- +(240:1) -- +(300:1) -- cycle;
        \draw[] (240:1) +(0:1) -- +(60:1) -- +(120:1) -- +(180:1) -- +(240:1) -- +(300:1) -- cycle;
        \fill (0:2) circle(0.1);
        \fill (120:2) circle(0.1);
        \fill (240:2) circle(0.1);
        \draw[->,thick] (0:2) -- (0:1.3);
        \draw[->,thick] (120:2) -- (120:1.3);
        \draw[->,thick] (240:2) -- (240:1.3);
    \end{tikzpicture}
    \caption{\label{fig:network}
        Small network with three hexagonal cells.
        The base stations are positioned at the corners and the users are uniformly distributed in the shaded circular area in the center.
    }
\end{figure}

\section{Results}
\label{sec:results}

In this section, we demonstrate the performance of the proposed OBE method in a multi-cell uplink scenario.
We consider a small network with three cells in a similar setup as in~\cite{bjornson_massive_2017}.
The base stations are positioned at the corners of the cells pointing towards the center of the network as depicted in Fig.~\ref{fig:network}.
The users are uniformly distributed in the circular area in the center of the network.
We analyze the performance in the uplink, when all users transmit with the same power.
We have a full pilot-reuse in this scenario, i.e., the number of channel accesses used for pilots coincides with the number of users per cell.

We use the spatial channel model from the 3GPP report in~\cite{3gpp_spatial_2014} to generate the covariance matrices.
Specifically we simulate the micro-cell scenario with a cell diameter of 500m.
We normalize the channel covariance matrices such that the SNR figure $\rho\ul$ denotes the per-antenna SNR of a user exactly in the center of the network in Fig.~\ref{fig:network}.
Since the noise variance is set to one, this means that for a user $k$ in the center we have $p_k\tr(\vC_k)/M = 1$.

To analyze the performance in the uplink, we use the bound on the spectral efficiency with the SINR in~\eqref{eq:uplink_sinr_cond}.
That is, we evaluate the expectation $\expec[\log_2(1+\tilde\gul_k)]$ for each user with Monte-Carlo simulations.
To demonstrate the asymptotic properties of our approach, we show the average spectral efficiency of the worst user in the cell with respect to the number of antennas in Fig.~\ref{fig:uplink_wrtM}.
Additional to the OBE, we show results for the LMMSE filter, the zero-forcing filter based on the MMSE channel estimates, and the matched filter based on the MMSE channel estimates.
As discussed in Section~\ref{sec:structure} and Appendix~\ref{app:channel_model}, in our scenario with a ULA at the base station, the covariance matrices are approximately diagonalized by the DFT matrix.
For the OBE and the LMMSE filter we also show results for the case where only the diagonals of the transformed channel covariance matrices are available (OBE-D and LMMSE-D in the legend).

We observe the expected behaviour.
The spectral efficiency saturates when using the simple MF based on the MMSE estimates. 
For all other methods it grows without bound.

The loss in performance for partial covariance matrix information is negligible.
Thus, in our scenario, there is no reason not to use the low-complexity variants that only use the diagonals of the covariance matrices.

\begin{figure}[t]
\begin{center}
        \begin{tikzpicture}[trim axis left,trim axis right]
\begin{semilogxaxis}[
    height=0.75\columnwidth,
    width=0.95\columnwidth,
    grid=both,
    ticks=both,
    xlabel={\small Number of antennas},
    ylabel={\small $\min_k r_k$},
    title={},
    ymin=0,
    ymax=4,
    xmin=10,
    xmax=1000,
    yticklabel style={/pgf/number format/.cd, fixed, fixed zerofill, precision=1},
    legend entries={
        \tiny LMMSE filter,
        \tiny LMMSE-D,
        \tiny OBE,
        \tiny OBE-D,
        \tiny MMSE zero-forcing,
        \tiny MMSE matched filter,
    },
    legend style={legend columns = 3, at={(0.00,1.05)}, anchor=south west},
]

\addplot[lmmse] table [x index=0,y index=6,col sep=comma] {figures/uplink_wrtM.csv};
\addplot[lmmsed] table [x index=0,y index=4,col sep=comma] {figures/uplink_wrtM.csv};
\addplot[rmf] table [x index=0,y index=5,col sep=comma] {figures/uplink_wrtM.csv};
\addplot[rmfd] table [x index=0,y index=3,col sep=comma] {figures/uplink_wrtM.csv};
\addplot[mmsezf] table [x index=0,y index=2,col sep=comma] {figures/uplink_wrtM.csv};
\addplot[mmsemf] table [x index=0,y index=1,col sep=comma] {figures/uplink_wrtM.csv};

\end{semilogxaxis}
\end{tikzpicture} 
\end{center}
    \caption{
        Average spectral efficiency of the worst user in the cell.
        Results are for the multi-cell scenario as depicted in Fig.~\ref{fig:network}, with $5$ users per cell and the same number of orthogonal training sequences.
        The cell-edge SNR is $-6$dB.
}
    \label{fig:uplink_wrtM}
\end{figure}
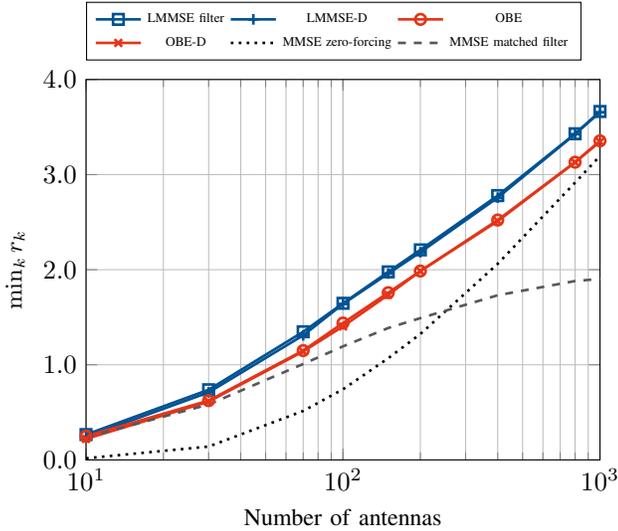

In Fig.~\ref{fig:uplink_wrtP}, we show the spectral efficiency with respect to the cell edge SNR.
We see that the gap between the LMMSE filter and the OBE vanishes for low SNR.
For high SNR, all methods saturate to different values.
Interestingly, zero-forcing saturates to a lower value than the proposed OBE.
The saturation of the spectral efficiency for the LMMSE and the zero-forcing filter is clearly due to pilot-contamination.
So while pilot contamination does not lead to an upper limit of the SINR with respect to the number of antennas, it does limit the SINR with respect to the transmit power.
Of course, there are other issues in practice, such as outdated CSI, that have a similar effect.

We see that at very high SNR the performance drops if we have only partial covariance matrix information.
We conclude that for lower noise variance and thus smaller regularization terms in the filter calculation, the inaccuracies in the covariance matrix information lead to increased interference.
In practice, this is not an issue, since the desired operation point of a massive MIMO system is at lower SNR.

\begin{figure}[t]
\begin{center}
        \begin{tikzpicture}[trim axis left,trim axis right]
\begin{axis}[
    height=0.75\columnwidth,
    width=0.95\columnwidth,
    grid=both,
    ticks=both,
    xlabel={\small $\pul$ in dB},
    ylabel={\small $\min_k r_k$},
    title={},
    ymin=0,
    ymax=4,
    xmin=-30,
    xmax=20,
    yticklabel style={/pgf/number format/.cd, fixed, fixed zerofill, precision=1},
    legend entries = {
        \tiny LMMSE filter,
        \tiny LMMSE-D,
        \tiny OBE,
        \tiny OBE-D,
        \tiny MMSE zero-forcing,
        \tiny MMSE matched filter,
    },
    legend style={legend columns = 3, at={(0.00,1.05)}, anchor=south west},
]

\addplot[lmmse] table [x index=0,y index=6,col sep=comma] {figures/uplink_wrtP.csv};
\addplot[lmmsed] table [x index=0,y index=4,col sep=comma] {figures/uplink_wrtP.csv};
\addplot[rmf] table [x index=0,y index=5,col sep=comma] {figures/uplink_wrtP.csv};
\addplot[rmfd] table [x index=0,y index=3,col sep=comma] {figures/uplink_wrtP.csv};
\addplot[mmsezf] table [x index=0,y index=2,col sep=comma] {figures/uplink_wrtP.csv};
\addplot[mmsemf] table [x index=0,y index=1,col sep=comma] {figures/uplink_wrtP.csv};

\end{axis}
\end{tikzpicture} 
\end{center}
    \caption{
        Average spectral efficiency of the worst user in the cell.
        Results are for the multi-cell scenario as depicted in Fig.~\ref{fig:network}, with $5$ users per cell and the same number of orthogonal training sequences.
        The number of antennas at the base station is $M=200$.
}
    \label{fig:uplink_wrtP}
\end{figure}
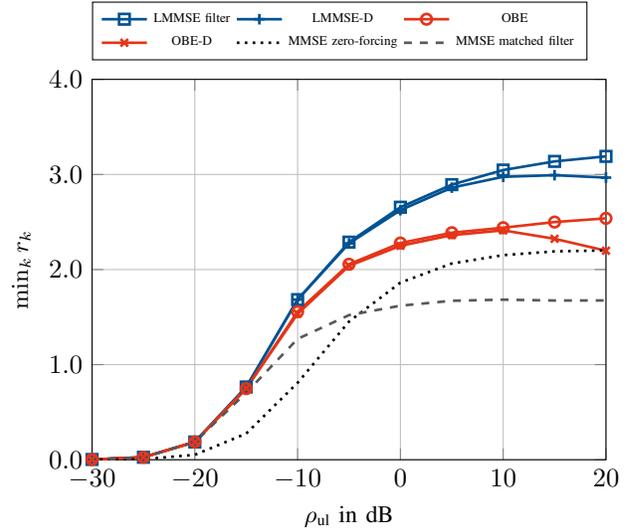

\section{Conclusion}

We presented a novel approach for the design of low-complexity linear receive filters with imperfect channel state information.
We get low complexity, since the receive filters themselves are linear in the observation obtained from a training phase.
Under mild assumptions on the covariance matrices, these bilinear equalizers (BE) yield an unbounded SINR for large numbers of antennas even in the presence of pilot-contamination.
Our simulation results indicate that in certain cellular scenarios, the BE approach is competitive with more complex methods such as the LMMSE filter.

In this work we focused on the optimal BE (OBE) design and its performance analysis.
We could show results analogously to those provided in~\cite{bjornson_massive_2017} for the LMMSE filter.
We want to note again that, due to the similarity of our SINR expressions to those for perfect CSI, we can apply state-of-the-art methods on downlink precoder design using uplink-downlink duality, resource allocation, QoS optimization and more.
The difference is that, since we optimize transformations that only depend on channel statistics, all of the resource allocation optimization problems also work on channel statistics and there is no need to perform complex optimizations in each channel coherence interval.

For these reasons, the OBE design is well suited for massive MIMO setups that have to deal with high path-loss, but also for the cell-free massive MIMO scenario, where we want to limit the amount of backhaul traffic.
Possibly, the design can also be adapted to the hybrid analog-digital transceiver structures discussed for millimeter wave bands.

\appendix

\subsection{SINR Reformulation}
\label{app:sinr_reformulation}

\begin{figure*}[t]
\begin{align}\label{eq:circulant_asy}
    [\vGamma_\text{ULA}]_{nk} &= \frac{1}{M}\tr \bigg(\tilde \vC(f_n) \Big(\sum_i p_i \tilde \vC(f_i) + \id \Big)\inv \tilde \vC(f_k) \Big( \sum_{i\in\Omega_p} \tilde \vC(f_i) + \frac{1}{\ptr} \id \Big)\inv \bigg) \notag \\
    &= \frac{1}{M} \sum_{m=0}^{M-1} \frac{f_n(2\pi m/M) f_k(2\pi m/M)}{ (\sum_i p_i f_i(2\pi m/M) + 1)( \sum_{i\in\Omega_p} f_i(2\pi m/M) + 1/\ptr) }.
\end{align}
\hrule
\end{figure*}

With the definitions before Lemma~\ref{thm:sinr_reformulation} we get
\begin{align}
    \va_k^\star &= \left( \vQ\tp\otimes\vZ + \vXi\vP\vXi\he - p_k \vc_k\vc_k\he \right) \inv \vc_k \\
                &= \left( \vUpsilon - p_k \vc_k \vc_k\he \right)\inv \vc_k.
\end{align}
where $\vUpsilon = \vQ\tp\otimes \vZ + \vXi\vP\vXi\he$.
We apply the matrix inversion lemma resulting in
\begin{align}
    \va_k^\star &= \vUpsilon\inv \vc_k - \vUpsilon\inv \vc_k ( \vc_k\he \vUpsilon\inv \vc_k - p_k\inv )\inv \vc_k\he\vUpsilon\inv\vc_k\\
          &= \vUpsilon\inv \vc_k \frac{1}{1 - p_k\vc_k\he \vUpsilon\inv \vc_k}.
\end{align}
Further, we have
\begin{align}\label{eq:opt_filt_low_comp}
    \tilde \va_k^\star &= \vUpsilon\inv \vc_k = \left(\vQ\tp\otimes\vZ + \vXi\vP\vXi\he \right)\inv\vXi \ve_k \notag \\
                       &= \frac{1}{p_k} (\vQ\tpi\otimes\vZ\inv) \vXi \left( \vXi\he(\vQ\tpi\otimes\vZ\inv)\vXi + \vP\inv \right)\inv \ve_k
\end{align}
leading to the SINR
\begin{align}\label{eq:alt_ul_sinr}
    \gamma_k\ul &= p_k \vc_k\he\va_k^\star = \frac{p_k \vc_k\he\tilde\va_k}{1- p_k\vc_k\he\tilde\va_k} \notag \\
        & =\frac{\ve_k\tp \vGamma (\frac{1}{M}\vP\inv + \vGamma)\inv \ve_k }{1 - \ve_k\tp \vGamma(\frac{1}{M}\vP\inv + \vGamma)\inv \ve_k} \notag\\
        & = M p_k \frac{\ve_k\tp \vGamma (\frac{1}{M}\vP\inv + \vGamma)\inv \ve_k }{\ve_k\tp (\frac{1}{M}\vP\inv + \vGamma)\inv \ve_k}.
\end{align}
Note that, with
\begin{align} 
    &[\tilde \va_1, \ldots \tilde \va_{K_p}] \notag  \\
    &= ( \vQ\tp\otimes \vZ  + \vXi \vP \vXi\he )\inv \vXi \notag \\
    &= (\vQ\tpi \otimes \vZ\inv) \vXi ( \vP\inv + \vXi\he(\vQ\tpi\otimes\vZ\inv)\vXi)\inv \vP\inv
\end{align}
we can calculate the scaled transformations of users using the same pilot sequence simultaneously.

\subsection{Multi-path Physical Channel Model}
\label{app:channel_model}

Most physical channel models for performance evaluation of wireless networks assume passive antenna elements in the farfield of the received signal.
The received signal is usually modelled as the superposition of signals impinging on the array from different angles.
The random phase shifts of signals that arrive from different angles are assumed to be independent.
For a general planar antenna placement, we get the covariance matrices
\begin{align}
    \vC_k = \beta_k \int_0^{2\pi} \va(\theta) \va(\theta)\he \eta_k(\theta) d\theta
\end{align}
where $\va(\theta)$ denotes the steering vector of the array from angle $\theta$ and $\eta_k(\theta)$ is the power density of the signals received from the different angles.
For example, $\eta_k(\theta)$ could be a mixture of Laplace distributions which describe different scatterer clusters~\cite{3gpp_spatial_2014}.
Note that for antennas with directivity, the antenna pattern is also included in $\eta_k(\theta)$ and assumed to be the same for each antenna.

For a uniform linear array with an antenna distance of half of the wavelength, we have
\begin{align}
    [\va(\theta)]_m = \exp(j \pi m \sin(\theta)).
\end{align}
Due to $\va(\theta) = \va(\pi - \theta)$ we can rewrite the integral as
\begin{align}
    \vC_k = \beta_k \int_{-\pi/2}^{\pi/2} \va(\theta) \va(\theta)\he (\eta_k(\theta) + \eta_k(\pi - \theta)) d\theta
\end{align}
where we assume for notational convenience that $p(\theta)$ is periodic with $2\pi$.
Let $\tilde \eta_k(\theta) = \eta_k(\theta) + \eta_k(\pi - \theta)$.
The $\vC_k$ are Toeplitz matrices with $[\vC_k]_{mn} = t_k[m-n]$ where
\begin{align}\label{eq:toeplitz_entries}
    t_k[m] = \beta_k \int_{-\pi/2}^{\pi/2} e^{-j\pi m \sin\theta} \tilde \eta_k(\theta)  d\theta.
\end{align}

If we assume that the power densities $\eta_k(\theta)$ are Riemann integrable, we can show that the limit $\lim_{M\rightarrow \infty} \vGamma$ exists.
We use results from~\cite{gray_toeplitz_2006} on the asymptotic equivalence of sequences of Toeplitz matrices to sequences of circulant matrices.
To this end, we substitute $\theta = \sin\inv(\omega/\pi)$ in \eqref{eq:toeplitz_entries} such that the $t_k[m]$ are the Fourier coefficients of a function $f_k(\omega)$.
We get
\begin{align}
    t_k[m] = \frac{1}{2\pi}\int_{-\pi}^{\pi} e^{-j m \omega} \tilde \eta_k(\sin\inv(\omega/\pi)) \frac{2\pi\beta_k }{\sqrt{\pi^2-\omega^2}} d\omega
\end{align}
and identify
\begin{align}
    f_k(\omega) = \frac{2\pi\beta_k }{\sqrt{\pi^2-\omega^2}} \tilde \eta_k(\sin\inv(\omega/\pi))
\end{align}
which is defined on $(-\pi, \pi)$. If we extend $f_k(\omega)$ periodically (defining $f(\pi+2n\pi) = 0, \forall n \in \mathbb Z$), we have
\begin{align}
    t_k[m] = \frac{1}{2\pi}\int_{0}^{2\pi} f_k(\omega) e^{-j m \omega} d\omega.
\end{align}
Now, applying the results from~\cite{gray_toeplitz_2006}, we define a circulant matrix $\widetilde \vC_k$ with the eigenvalues $f_k(2\pi m/M)$, $m=0,$ \dots, $M-1$ which is asymptotically equivalent to $\vC_k$.
Note that since $\widetilde \vC_k$ is circulant, the eigenvectors are given by the columns of the DFT matrix.

The asymptotic equivalence $\widetilde \vC_k \asymp_w \vC_k$ allows us to replace the covariance matrices in the trace expressions for the entries $[\vGamma]_{nk}$ in~\eqref{eq:gamma_elementwise} with the corresponding circulant matrices in the limit of $M$ going to infinity.
That is,
\begin{align}
    \lim_{M\rightarrow \infty} [\vGamma]_{nk} - [\vGamma_\text{ULA}]_{nk} = 0
\end{align}
where the entries of $\vGamma_\text{ULA}$ are given in~\eqref{eq:circulant_asy} at the top of this page.
The limit evaluates to
\begin{multline*}
    \lim_{M\rightarrow \infty} [\vGamma_\text{ULA}]_{nk} \\= \frac{1}{2\pi} \int_0^{2\pi} \frac{f_n(\omega) f_k(\omega)}{ (\sum_i p_i f_i(\omega) + 1)( \sum_{i\in\Omega_p} f_i(\omega) + 1/\ptr) } d\omega.
\end{multline*}

In other words, for large $M$ the covariance matrices $\vC_k$ can be replaced by the underlying spectra $f_k(\omega)$ in our calculations.
Since the value of $f_k(\omega)$ goes to infinity as $\omega \rightarrow \pm \pi$, the covariance matrices do not generally have bounded spectral norm even if the $\eta_k(\theta)$ are bounded (cf.~Condition~\ref{cond:bounded_norm}).
However, since the improper integral exists, this is not an issue.

In fact, we can resubstitute $\omega = \pi\sin(\theta)$ to get 
\begin{align}
    \lim_{M\rightarrow \infty} [\vGamma_\text{ULA}]_{nk} = \frac{1}{2} \int_{-\pi/2}^{\pi/2} \tilde \eta_n(\theta) \tilde \eta_k(\theta) \alpha(\theta) d\theta.
\end{align}
where
\begin{align}
    \alpha(\theta) = \frac{\cos(\theta)}{ (\sum_k p_k \tilde \eta_k(\theta) + \cos(\theta))( \sum_{k\in\Omega_p} \tilde \eta_k(\theta) + \cos(\theta)/\ptr) }.
\end{align}
Since $\alpha(\theta)$ is essentially non-zero for all bounded $\tilde\eta_k(\theta)$, linear independence of the densities $\tilde \eta_k(\theta)$ leads to linear independence of $\sqrt{\alpha(\theta)} \tilde\eta_k(\theta)$.
That is, we need linearly independent power densities $\tilde\eta_k(\theta)$ for users that employ the same pilot sequence to get a linear scaling of the SINR.

We also see, that if the power densities of users that employ the same pilot sequence have non-overlapping support, $\vGamma$ converges towards a diagonal matrix, i.e., pilot-contamination has no impact on the asymptotically equivalent SINR in this case.
This effect, namely that for non-overlapping angular support covariance matrices of different users are asymptotically orthogonal, was also observed and exploited in previous work~\cite{adhikary_joint_2013}.

\subsection{LMMSE Filter}
\label{app:lmmse}

In this section, we discuss the connection of the LMMSE filter to the BE receiver.
The LMMSE filter was used, e.g., in~\cite{bjornson_massive_2017} and~\cite{neumann_mse_2017} to analyze the performance for a large number of antennas.
This filter minimizes the MSE $\expec[\lvert s_k - \vg_k\he \vy \vert \obs \rvert^2]$, but at the same time it maximizes a lower-bound on the mutual information.
For the LMMSE filter, we get a lot of quantities that are analogous to those introduced in Section~\ref{sec:uplink} for the BE approach.
In the following, we put a tilde on top of the corresponding symbols to denote the equivalence.
For example, we have the optimal SINR $\gamma_k^\star$ for the BE and an analogous optimal SINR $\tilde \gamma_k^\star$ for the LMMSE.

We can use exactly the same technique as in Section~\ref{sec:uplink} to formulate a lower bound on the \emph{conditional} mutual information 
\begin{align}
    I(s_k;\hat s_k \vert \obs) \geq \log_2( 1 + \tilde\gamma_k\ul)
\end{align}
where $\obs$ denotes the set of all observations from the training phase.
We simply have to replace all expectations in~\eqref{eq:sinr_lower_bound} by conditional expectations leading to
\begin{align}
    \tilde\gamma_k\ul &= \frac{ p_k\abs{\expec[\vg_k\he\vh_k\vert\obs]}^2}{ \expec[\vg_k\he\vg_k\vert \obs] + p_k \var(\vg_k\he\vh_k\vert \obs) + \sum_{n\neq k} p_n \expec[\abs{\vg_k\he\vh_n}^2\vert \obs] }.
\end{align}
Since the filters $\vg_k$ are deterministic functions of the observations $\obs$, this can be simplified to
\begin{align}\label{eq:uplink_sinr_cond}
    \tilde\gamma_k\ul &= \frac{ p_k\abs{\vg_k\he\hest_k}^2}{ \vg_k\he\vg_k + p_k \vg_k\he \widetilde\vC_k \vg_k + \sum_{n\neq k} p_n \vg_k\he(\widetilde\vC_n + \hest_n\hest_n\he) \vg_k }.
\end{align}
with the MMSE estimates of the channel vector $\hest_k = \expec[\vh_k\vert \obs]$ and the covariance matrices of the corresponding estimation errors $\widetilde \vC_k = \expec[ (\hest_k - \vh_k)(\hest_k-\vh_k)\he]$.

By averaging the lower-bound on the conditional mutual information, we get another lower bound on the mutual information, i.e.,
\begin{align}
    I(s_k;\hat s_k) = \expec[ I(s_k;\hat s_k \vert \obs) ] \geq \expec[\log_2(1+ \tilde \gamma_k\ul)]
\end{align}
which is tighter than the lower bound we used in the Section~\ref{sec:uplink}, but does not yield a convenient analytical expression when we incorporate the BE approach.

The SINR in~\eqref{eq:uplink_sinr_cond} is maximized by the filter 
\begin{align}
    \vg_k^\star = \left( \id + \sum_n p_n\widetilde \vC_n + \sum_{n \neq k} p_n\hest_n\hest_n\he  \right)\inv \hest_k
\end{align}
leading to the optimal SINR
\begin{align}\label{eq:uplink_sinr_cond_opt}
    \tilde \gamma_k^\star = p_k \hest_k\he \left( \id + \sum_n p_n\widetilde \vC_n + \sum_{n\neq k} p_n\hest_n\hest_n\he  \right)\inv \hest_k.
\end{align}
The filter $\vg_k^\star$ differs from the LMMSE filter~\cite{bjornson_massive_2017,neumann_mse_2017}
\begin{align}
    \vg_k^\text{LMMSE} = \left( \id + \sum_n p_n\widetilde \vC_n + \sum_{n} p_n\hest_n\hest_n\he  \right)\inv \hest_k
\end{align}
only in a scaling factor, which arises due to the fact that we now sum over the channel estimates of all users inside the inverse.
Consequently, the LMMSE filter also maximizes the SINR in~\eqref{eq:uplink_sinr_cond}.

We use 
\begin{align}
    \widetilde \vZ = \id + \sum_n p_n \widetilde \vC_n \quad\text{and}\quad \Hest = [\hest_1,\ldots, \hest_K]
\end{align}
to define 
\begin{align}
    \vGamma^\obs = \Hest\he \widetilde\vZ\inv \Hest/M.
\end{align}
We can use the same steps as in Lemma~\ref{thm:sinr_reformulation} to reformulate the SINR in~\eqref{eq:uplink_sinr_cond_opt} to
\begin{align}\label{eq:uplink_sinr_cond_reform}
    \tilde \gamma_k^\star = M p_k \frac{\ve_k\tp \vGamma^\obs (\frac{1}{M}\vP\inv + \vGamma^\obs)\inv \ve_k }{\ve_k\tp (\frac{1}{M}\vP\inv + \vGamma^\obs)\inv \ve_k}.
\end{align}

With Conditions~\ref{cond:linear_independence} and~\ref{cond:bounded_norm} on the channel covariance matrices, this SINR has been shown to grow without bound in~\cite{bjornson_massive_2017}.
A convenient form of an asymptotically equivalent SINR can be derived along the lines of the derivation of the asymptotically equivalent MSE in~\cite{neumann_mse_2017}.

With Condition~\ref{cond:bounded_norm}, the entries of $\vGamma^\obs$ are asymptotically equivalent to their expectations.
That is,
\begin{align}
    \vGamma^\obs \asymp \widetilde \vGamma = \expec[\vGamma^\obs].
\end{align}

For users $k$ and $n$ that use a different pilot sequence we get $[\widetilde \vGamma]_{kn} = 0$.
Thus for proper arrangement of the user indices, the matrix $\widetilde \vGamma = \blkdiag(\widetilde\vGamma_1,\ldots,\widetilde\vGamma_{\Ttr})$ has block-diagonal structure with one block $\widetilde \vGamma_p$ for each pilot sequence.
We consider again the set of users $\Omega_p = \{1, \ldots, K_p\}$ which use the same pilot sequence $p$.
For this pilot sequence the corresponding block of $\widetilde \vGamma$ is given by
\begin{align}\label{eq:gamma_p}
    \widetilde \vGamma_p = \vXi\he (\vQ\inv \otimes \widetilde \vZ\inv) \vXi/M
\end{align}
Thus, we get the deterministic, asymptotically equivalent SINR
\begin{equation}\label{eq:uplink_sinr_cond_det}
    \tilde\gamma_k^\text{det} = M p_k \frac{\ve_k\tp \widetilde\vGamma_p (\frac{1}{M}\vP\inv + \widetilde\vGamma_p)\inv \ve_k }{\ve_k\tp (\frac{1}{M}\vP\inv + \widetilde\vGamma_p)\inv \ve_k},
\end{equation}
which differs from the SINR in~\eqref{eq:uplink_sinr_gmf_reform} only in the fact that the $\vZ$ in $\vGamma$ is replaced by $\widetilde \vZ$ in $\widetilde \vGamma_p$.
Note that this SINR only depends on the channel statistics and not on the instantaneous observations, which is not only useful for performance analysis but also for resource allocation in a practical system, e.g., power allocation and pilot sequence allocation.

Since $\widetilde \vZ$ contains the covariance matrices of the estimation errors instead of the channel covariance matrices we have $\widetilde \vZ \prec \vZ$.
From the formulation of the SINR in~\eqref{eq:uplink_sinr_gmf_opt}, it is clear that $\widetilde \vZ \prec \vZ$ leads to $\widetilde \gamma_k^\text{det} > \gamma_k^\star$ but also to $\widetilde \vGamma_p \succ \vGamma$.
Consequently, we can define another asymptotically equivalent SINR analogously to~\eqref{eq:uplink_sinr_asy}
\begin{align}\label{eq:uplink_sinr_cond_asy}
    \tilde \gamma_k^\text{asy} = M \frac{p_k}{\ve_k\tp \widetilde\vGamma_p\inv \ve_k}
\end{align}
and know from Lemma~\ref{lemma:general_conditions} that
\begin{align}
    \tilde \gamma_k^\star \asymp \tilde \gamma_k^\text{det} \asymp \tilde \gamma_k^\text{asy}
\end{align}
with $\liminf \tilde \gamma_k^\text{asy}/M > 0$.

Note that $\widetilde \vZ \not\asymp_w \vZ$ and thus
\begin{align}
    \liminf_M \tilde \gamma_k^\text{asy} / \gamma_k^\text{asy} > 1.
\end{align}
If the limit $\lim_M \tilde \gamma_k^\text{asy} / \gamma_k^\text{asy}$ exists, it gives us an idea about how many more antennas are needed for the BE approach to get a similar performance to the MMSE filter in the large array regime.  
However, we used different bounds on the mutual information to derive the asymptotic SINRs for the different approaches so the comparison is not entirely fair.
Unfortunately, we do not get convenient asymptotic expressions when we incorporate the MMSE filter in the bound in~\eqref{eq:sinr_lower_bound} or the BE in the SINR in~\eqref{eq:uplink_sinr_cond}.
Thus, the actual performance difference in practice has to be evaluated by monte-carlo simulations.

Nevertheless, there are scenarios where $\gamma_k^\star$ is close to $\tilde\gamma_k^\star$.
Clearly, for low SNR, $\vZ$ is similar to $\widetilde \vZ$.
Low SNR does not necessarily mean low SINR if we compensate for the reduced SNR by increasing the number of antennas.
We can show the following result.
\begin{thm}
    If the $p_k$ and $\ptr$ go to zero with $M$, such that $\lambda_k = M p_k \ptr = \order(1)$, we get $\gamma_k^\star \asymp \tilde\gamma_k^\star \asymp \gamma_k^\text{low}$ with
    \begin{align}\label{eq:sinr_low_snr}
        \gamma_k^\text{low} = \frac{\ve_k\tp \vLambda \vXi\he\vXi/M \left( \id + \vLambda \vXi\he\vXi/M \right)\inv \ve_k}{\ve_k\tp \left( \id + \vLambda\vXi\he\vXi/M \right)\inv \ve_k}.
    \end{align}
    where $\vLambda = \diag(\lambda_1,\ldots, \lambda_{K_p})$.
\end{thm}
\begin{proof}
    We can reformulate
    \begin{align}
        \tilde \gamma_k^\star = \frac{\ve_k\tp \vLambda \frac{1}{\ptr}\vGamma^\obs (\id + \vLambda \frac{1}{\ptr}\vGamma^\obs)\inv \ve_k }{\ve_k\tp ( \id + \vLambda \frac{1}{\ptr}\vGamma^\obs)\inv \ve_k}
    \end{align}
    and
    \begin{equation}
        \gamma_k^\star = \frac{\ve_k\tp \vLambda \frac{1}{\ptr}\vGamma(\id + \vLambda \frac{1}{\ptr}\vGamma)\inv \ve_k }{\ve_k\tp ( \id + \vLambda \frac{1}{\ptr}\vGamma)\inv \ve_k}.
    \end{equation}
    For the given assumptions, we have $\vZ \asymp_w \widetilde \vZ \asymp_w \id$ and $\vQ_k /\ptr \asymp_w \id$, which leads to $\vGamma/\ptr \asymp \tilde\vGamma_p/\ptr \asymp \vXi\he\vXi/M$ which leads to the desired result.
\end{proof}
We control the SINR in~\eqref{eq:sinr_low_snr} by changing the $\lambda_k$.
Note that the SINR does not go to zero when the covariance matrices are linearly dependent.
However, if they are, the SINR $\gamma_k^\text{low}$ will saturate for large $\lambda_k$.
On the other hand, if the covariance matrices are independent, we can achieve any combination of SINRs with corresponding $\lambda_k$.

We found that the SINR expression in~\eqref{eq:sinr_low_snr} does not lead to a good approximation of the achievable rates for practical numbers for the system parameters.
However, we can see from the numerical results in Section~\ref{sec:results} that, indeed, the performance gap between the LMMSE filter and the BE vanishes for low SNR.

\bibliographystyle{IEEEtran}
\bibliography{IEEEabrv,literature}

\end{document}